%% file: HICSS.tex
\documentclass[10pt]{article}
\pdfoutput=1

\usepackage[normalem]{ulem}

\usepackage[letterpaper]{geometry}
\usepackage{hicss}
\usepackage{times}
\usepackage[none]{hyphenat}
\usepackage{url}
\usepackage{latexsym}
\usepackage{graphicx}
\usepackage{fancyhdr}
\graphicspath{{images/}}
\usepackage[
    style=apa,
  ]{biblatex}
\usepackage[colorlinks]{hyperref}

\usepackage{amsmath,amsfonts,amstext,amsxtra,amssymb,mathtools,amsthm}
\numberwithin{equation}{section}
\usepackage{csquotes}
\usepackage[expansion=true,protrusion=true]{microtype}
\usepackage{textcomp}
\usepackage{enumitem} 
	\makeatletter % for cross-referencing description items, see https://texblog.org/2012/03/21/cross-referencing-list-items/
	\def\namedlabel#1#2{\begingroup
		#2%
		\def\@currentlabel{#2}%
		\phantomsection\label{#1}\endgroup
	}
	\makeatother
\usepackage{blindtext}
\usepackage{hyperref}
\hypersetup{colorlinks=false}
\usepackage{booktabs}%extra lines in tabulars
\usepackage{cleveref} 
\usepackage{xcolor}
\usepackage{todonotes}
\usepackage{upgreek}
\usepackage{dsfont, bbold}
\usepackage{nicefrac}
\usepackage{esdiff}
\usepackage{physics}
\usepackage{tikz}
\usepackage{float}
\usetikzlibrary{quantikz}
\input{SYSTEM/commands.tex}

\input{SYSTEM/FJSP.tikzstyle}
\title{Graph-controlled Permutation Mixers in QAOA for the Flexible Job-Shop Problem}

\fancypagestyle{specialfooter}{%
  \fancyhf{}
  
  \fancyfoot[R]{ \noindent\fbox{%
    \parbox{16.2 cm}{%
        {\footnotesize This work has been submitted to the HICSS 2024 for possible publication. Copyright may be transferred without notice, after which this version may no longer be accessible.}
        }
    }}
}

\author{Lilly Palackal$^{\ast}$ \\
  Infineon Technologies AG \\
  Technical University of Munich \\
 {\underline{lilly.palackal@infineon.com}} \\ \\
   \And
  Leonhard Richter$^{\ast}$ \\
  Friedrich-Alexander University \\
  Erlangen-Nürnberg \\
  {\underline{leonhard.richter@outlook.de} } \\  
  \\
  {\footnotesize $^{\ast}$Authors contributed equally}\\ \\
   \And
  Maximilian Hess \\
  Infineon Technologies AG \\
  Leibniz University Hannover \\
  {\underline{maximilian.hess@infineon.com} } \\ \\
   }

\date{}

\begin{document}
\maketitle
\thispagestyle{specialfooter} % Only for the arxiv pre-print
\input{HICSScontent/hicss_abstract}
\input{HICSScontent/hicss_text}

\input{HICSScontent/hicss_acknowl}

\vspace{2cm}

% if added before the last page, this command can help balancing columns
%\addtolength{\textheight}{-.2cm} 

%Bibliography 
%\bibliographystyle{apalike}
% \bibliography{sample}
%\input{HICSS.bbl}
\printbibliography

\end{document}

%% file: SYSTEM/commands.tex
\newcommand{\qwbund}{\qwbundle[alternate]{}}

\newcommand{\feas}{\mathcal{F}}
\newcommand{\mix}{B}

\DeclareMathOperator{\Span}{span}

\newcommand{\qaoa}{\textsc{CM-QAOA}}
\newcommand{\fjsp}{\textsc{FJSP}}

\newcommand{\of}[1]{\left(#1\right)}
\newcommand{\cof}[1]{\left\{#1\right\}}

\let\abs\vof

\newcommand{\imply}{\; \Rightarrow \;}
\newcommand{\mDef}{\vcentcolon =}

\newcommand{\N}{\mathbb{N}} %natural numbers
\newcommand{\R}{\mathbb{R}} %reals
\newcommand{\C}{\mathbb{C}} %complex numbers

\newcommand{\OneOp}{\mathbb{1}}
\newcommand{\imun}{i}

\newcommand{\primed}{^\prime}

\newcommand{\inv}{^{-1}}

\DeclarePairedDelimiterX\Set[1]{\lbrace}{\rbrace}%
{  #1 }
\DeclareMathOperator{\nbhd}{nbhd}

 \DeclarePairedDelimiterX\set[1]\lbrace\rbrace{\setaux#1}
 \def\setaux#1|{#1\;\delimsize\vert\;}
 
\theoremstyle{definition}
\newtheorem{definition}{Definition}

\newtheorem{theorem}{Theorem}

\newtheorem{lemma}{Lemma}

%% file: HICSScontent/hicss_abstract.tex
\begin{abstract}
    One of the most promising attempts towards solving optimization problems with quantum computers in the noisy intermediate scale era of quantum computing are variational quantum algorithms.
The Quantum Alternating Operator Ansatz provides an algorithmic framework for constrained, combinatorial optimization problems.
As opposed to the better known standard QAOA protocol, the constraints of the optimization problem are built into the mixing layers of the ansatz circuit, thereby limiting the search to the much smaller Hilbert space of feasible solutions.
In this work we develop mixing operators for a wide range of scheduling problems including the flexible job shop problem.
These mixing operators are based on a special control scheme defined by a constraint graph model.
After describing an explicit construction of those mixing operators, they are proven to be feasibility preserving, as well as exploring the feasible subspace.

\end{abstract}

\subsubsection*{Keywords:}

Quantum Computing, Quantum Algorithms, Constraint-Mixer QAOA, Job-Shop Scheduling

%% file: HICSScontent/hicss_text.tex
\section{Introduction}

    The usage of quantum mechanical phenomena in computing provably enables asymptotic speedups over classical computers in some computational tasks (\cite{Shor_1997}, \cite{grover1996fast}), and researchers are aiming to discover such algorithmic breakthroughs in others.
    The reality of quantum computing right now, sometimes called the noisy intermediate scale quantum (NISQ) era, presents additional challenges: While quantum processors are only available with a modest number of noisy qubits, the algorithms under consideration must address these shortcomings. In the case of quantum optimization, hopes for near- to mid-term algorithmic advances rest on the shoulders of variational hybrid algorithms (\cite{Cerezo_2021}). 
    All of them involve a parameterized quantum circuit $U(\theta)$ which is used to generate an ansatz state 
    \begin{equation}
        \ket{\psi(\theta)} := U(\theta)\ket{0}.
    \end{equation} 
    Measuring this ansatz state several times in the computational basis provides solution candidates (the bit string representations of the computational basis states) to the (binary) optimization problem. Based on the quality of these samples with respect to the optimization problem, the parameters $\theta$ are then updated and fed back into the quantum circuit.
    In this algorithmic scheme, one qubit roughly corresponds to one binary variable in the optimization problem. Without the need for spatial overhead, such as auxiliary quantum registers, relevant problem sizes can be targeted much sooner, thereby adressing the ``intermediate scale" concern. 
    Due to the variational nature of the algorithm, where measurement outcomes inform the choice of parameters and thus the eventual solution, noise is taken into account and filtered out to a certain extent. This is in contrast to pure quantum algorithms, where each step ``trusts" the preceding steps to prepare a certain quantum state, thus allowing errors to propagate easily.\\
    The Quantum Alternating Operator Ansatz \parencite{hadfield_diss} also referred to as the constraint mixer Quantum Approximate Optimization Algorithm (\qaoa{}) \parencite{CM_QAOA_Fuchs2022} is a framework for such algorithms. Its setup is very similar to the better known Quantum Approximate Optimization Algorithm \parencite{farhi_qaoa}, involving alternating mixing and phase separation layers. 
    The main distinction point of \qaoa{} is to explore only feasible solutions for a given optimization problem, thus drastically reducing the dimension of the search space. This is in contrast to many other variational algorithms like standard QAOA, where infeasible solutions are assigned undesirably high values via penalty terms and the mixing layers create a superposition of all bit strings.
    In this work, we introduce a concrete construction following this scheme, which is suitable for a possibly wide range of scheduling problems including the \emph{Flexible Job-Shop Problem} (\fjsp{}).
    To this end we consider a \emph{constraint-graph model} that encodes the constraints of \fjsp{}-instances in a graph, similarly to the work of \cite{hannover_paper}.
    We use the constraint graph to define a unitary mixing operator for the \fjsp{} which is made up of controlled permutation operators.
    We prove that this mixing operator is both feasibility preserving and explores the feasible subspace.\\
    In \Cref{sec:preliminaries} we formally introduce the two main ingredients of this work, namely the algorithmic framework CM-QAOA in \Cref{subsec:qaoa} and the flexible job shop problem in \Cref{subsec:fjsp}. Section \ref{sec:cm-operators} goes on to construct the constraint graph of the \fjsp{} (\Cref{subsec:fjsp-constraint-graph}). In \Cref{subsec:graph-controlled permuation mixers} we go on to introduce the concept of controlled permutation unitaries and provide an explicit construction of a mixing operator for the \fjsp{}. We prove that this operator is feasibility-preserving in \Cref{thm:feasibility} and that it explores the feasible subspace in \Cref{thm:explorability}.
    Finally, in \Cref{sec:discussion}, the results are compared to related work and the practical limitations of our approach are discussed.
    
\section{Preliminaries}\label{sec:preliminaries}
\subsection{The Quantum Alternating Operator Ansatz}\label{subsec:qaoa}
    The \qaoa{}, originally proposed by \textcite{hadfield_diss}, is a framework of variational quantum algorithms generalizing the Quantum Approximate Optimization Algorithm \parencite{farhi_qaoa}.
    The \qaoa{} allows for a restriction of the search space to the subspace of feasible states, i.e. computational basis states whose bit string representations correspond to feasible solutions to an optimization problem. More precisely, the subspace of feasible solutions is spanned by computational basis states corresponding to bit strings which satisfy all constraints of the given problem instance.
    This potentially decreases the size of the state space in which the goal function is to be optimized substantially.
    To achieve this, the main effort in algorithm design is shifted from constructing a cost Hamiltonian representing the goal function as well as the constraints of a problem instance towards the implementation of the constraints directly into mixing operators. A combinatorial optimization problem which can be solved by \qaoa{} is given by a pair $\of{F,f}$, where, for $N\in\N$, $f:\cof{0,1}^N\to \R$ is a cost function to be minimized\footnote{As maximization of $f$ is the same as minimization of $-f$, it is sufficient to only consider minimization problems.} and 
    \begin{equation}\label{eq:feas_classical}
        F\subset \cof{0,1}^N
    \end{equation}
    is the set of feasible solutions typically characterized by a list of constraints. In the \qaoa{} the choice of alternating operators is more flexible compared to the QAOA. Especially, the mixing operators no longer need to be generated by a mixer Hamiltonian as in the work of \cite{farhi_qaoa}, but is rather required to fulfill certain \emph{design criteria} \cite{hadfield_diss}. The substitute for the QAOA-mixer $\sum_i X_i$ must have two features which ensure that the search of \qaoa{} remains within the feasible subspace $\feas$ and at the same time every feasible state can be reached for a suitable choice of parameters.

    \begin{definition}[Design criteria for mixing operators]\label{def:feas_reach}
        Let the tuple $(F \subset \{0,1\}^N, f)$ define a constrained optimization problem and let $B_{\beta} = \set*{\mix\of{\beta}|\beta\in\R}$ be a family of unitaries acting on the Hilbert space $\mathcal{H} \coloneqq (\C^2)^{\otimes N}$. Denote the subspace of feasible solutions by $\feas \coloneqq \Span(\{\ket{x} | x \in F\}) \subset \mathcal{H}$.
        The family $B_\beta$ is said to \dots
        \begin{itemize}
            \item[\dots] \emph{preserve feasibility}, if and only if
                $$\forall\beta\in\R\colon\:\mix\of{\feas}\mDef\set*{\mix(\beta)\ket{\psi}|\ket{\psi}\in\feas}\subseteq\feas.$$
            \item[\dots] \emph{explore the feasible subspace}, if and only if for all $x,y\in F\colon$ there exists a $\beta^{\ast}\in\R\qc$ and $r\in\N$ s.t.
            $$\abs{
                \bra{y}\mix^r\of{\beta^\ast}\ket{x}
            }>0.$$
        \end{itemize}
    \end{definition}
    \begin{definition}[\qaoa{}]\label{def:qaoa}
        
        A \qaoa{} instance is given by two families of unitary operators, $B_{\beta} = \set*{\mix\of{\beta} | \beta\in\R}$ and $C_{\gamma} = \set*{e^{-\imun \gamma H_f}|\gamma\in\R}$, and a feasible initial state $\ket{\psi_0}\in\feas$, such that each operator in $B_{\beta}$ fulfills the criteria in \Cref{def:feas_reach}.
        Here $H_f = \sum_{x \in \{0,1\}^N} f(x) \ketbra{x}$ is the Hamiltonian representing the function $f$ diagonally.\\
        The operator $\mix(\beta)$ is referred to as the \emph{mixing operator} and the operator $e^{-\imun \gamma H_f}$ is called \emph{phase-separation operator}.
    \end{definition}
    
\subsection{The Flexible Job-Shop Problem}\label{subsec:fjsp}  
    An example of a combinatorial optimization problem, which is relevant to industrial applications, is the \emph{Flexible Job-Shop Problem} (\fjsp{}).
    A number of jobs, each of which is given by a sequence of operations of varying processing times, have to be scheduled on a number of machines with the goal of optimizing a performance indicator, e.g. minimizing the overall execution time, often called \emph{makespan}.
    More concretely, an instance of \fjsp{} comprises the following data:
    \begin{itemize}
        \item $n_J$ jobs $J_i\in J=\cof{J_1,\dots,J_{n_J}}$, for some $n_J\in \N$,
        \item For each $j\in [n_J]$, a set of $p_j$ operations $O_{j,o}$ in $O_j=\cof{O_{j,1},\dots,O_{j,p_j}}$, for some $p_j\in \N$
        \item $n_M$ machines $M_m$ in $M=\cof{M_1,\dots,M_{n_M}}$, for some $n_M\in\N$,
        \item time-slots $T_t$ in $T=\cof{T_1,\dots,T_{n_T}}$, for some number of time steps $n_T\in\N$.
    \end{itemize}
    Each job $J_j$ consists of the operations $O_j$, that shall be assigned to a machine $M_m$ in $M$ at a given time-slot $T_t\in T$.
    Denoting the set of all operations by $O = \bigcup_{j,o} O_{j,o}$, an \emph{assignment} of operation $O_{j,o}$ is defined to be a tuple $(O_{j,o},M_m,T_t)\in O\times M\times T$.
    A subset $S\subseteq O\times M\times T$ of assignments, is called a \emph{schedule}.
The following additional data are needed to define feasible schedules.
    \begin{itemize}
    
        \item For each $O_{j,o}\in O$ a set $M_{j,o}\subseteq M$ of machines the operation may be executed on.
        \item For each $O_{j,o}\in O$ and $M_m\in M$, a duration $d_{j,o,m}\in N$ the operation $O_{j,o}$ takes, when being executed on machine $M_m$.
                If $M_m\notin M_{j,o}$ the corresponding duration does not matter and is set to infinity $d_{j,o,m}=\infty$.
    \end{itemize}

    A schedule is called feasible if it satisfies the following constraints:
    \begin{enumerate}
        \item\label{itm:ass_constr}\textbf{Assignment constraint} For every operation $O_{j,o}\in O$, there is exactly one time-slot $T_t\in T$ and exactly one machine $M_m\in M$, such that $\of{O_{j,o},M_m,T_t}\in S$. In other words, $\of{O_{j,o},M_m,T_t}\in S \Rightarrow \of{O_{j,o},M_{m'},T_{t'}} \notin S$ if $m \neq m'$ and $t \neq t'$.
        \item\label{itm:ord_constraint} \textbf{Order constraint} For all assignments $\of{O_{j,o},M_m,T_t}$, $\of{O_{j,o\primed},M_{m\primed},T_{t\primed}}\in S$ with $o<o\primed$, operation $O_{j,o\primed}$ is started only after operation $O_{j,o}$ has been finished, i.e. $t + d_{j,o,m} \leq t\primed$.
        \item \label{itm:mach_constr} \textbf{Machine constraint} For every machine $M_m\in M$, operations $O_{j,o} \neq O_{j\primed,o\primed} \in O$ and $T_t,T_{t\primed} \in T$ with $\of{O_{j,o},M_m,T_t},\of{O_{j\primed,o\primed},M_m,T_{t\primed}} \in S$ we have $t\primed \notin [t,t+d_{j,o,m})$. 
        That is, at no point in time two operations are being processed simultaneously on the same machine.
    \end{enumerate}
    A schedule is called \emph{feasible} if each of the constraints is satisfied.
    The number of time steps $n_T$ must be chosen such that it is guaranteed that a feasible schedule finishing within $n_T$ steps exists. This is usually achieved by constructing a feasible schedule via some computationally cheap heuristic method. 
    The aim of the \fjsp{} is to find a feasible schedule, which optimizes a cost function.
    Cost functions for the \fjsp{} can be defined in various fashions \parencite{fjsp_review}. One typical choice which will be our choice in this work is the makespan $C(S)\mDef\max\set*{T_t + d_{j,o,m} | \of{O_{j,o},M_m,T_t}\in S}$, i.e. the total processing time.
    \newline
    
    We formulate the \fjsp{} as a binary combinatorial optimization in the sense of \cref{eq:feas_classical} as follows.
    Every possible schedule $S\subseteq O\times M\times T$ (feasible or not) can be encoded by some binary string $x=x_1\dots x_N\in \cof{0,1}^N$, where $N\mDef \abs{O\times M\times T}$.
    Each $x_i$ denotes whether a given assignment is part of the schedule $S$.
    More precisely for a given schedule $S$, set
    \begin{equation}
        x_i \mDef \begin{cases}
            1, & \iota^{-1}(i)\in S\\
            0, & \text{else},
        \end{cases}
    \end{equation}
    where $\iota\colon\: O\times M\times T \to \cof{1,\dots, N}$ is a bijective map, enumerating all possible assignments.
    Thus $\iota$ induces a bijection $I\colon\:\mathcal{P}\of{O\times M\times T} \to \cof{0,1}^N$.
    The set $F\subseteq \cof{0,1}^N$ of feasible binary strings is then the set of binary strings $x\in \cof{0,1}^N$, 
    such that $S\of{x} \mDef I^{-1}\of{x}$ is a feasible schedule.
    Together with the cost function $f(x)\mDef C\of{\iota^{-1}(x)}$, $(F,f)$ is a binary combinatorial optimization problem as above.
    In particular, the bit $x_{\iota(O_o,M_m,T_t)}$ is set to 1 if we assign operation $o$ to machine $m$ at time $t$.

\section{Graph-based Constraint Mixer Operators}\label{sec:cm-operators}
    Many scheduling problems may be formulated as some variation of a graph-coloring problem \parencite{Marx2004GRAPHCP}.
    The common idea in all such approaches is to consider a graph, where each vertex corresponds to one object (job, assignment of a job to a machine,\dots), and the set of edges is constructed such that ``conflicting" objects share an edge. 
    Roughly speaking, the solution of the scheduling problem then becomes equivalent to finding some coloring of the corresponding graph \parencite{Marx2004GRAPHCP}.
    In many cases, each color represents a time slot or machine for the object to be assigned to.
    However, this approach is hard to realize for the \fjsp{}, because the available assignments for each job depend on the assignments of all the operations.
    Nevertheless, the \fjsp{} can be translated to a constraint graph:
    In contrast to other scheduling problems, not each job but each possible assignment $(j,o,m,t)$ corresponds to a vertex of the graph. 
    Then each pair of vertices, whose corresponding assignments are in conflict with each other, are connected by an edge.
    The resulting graph $G=(V,E)$ now encodes all the constraints of a given \fjsp{} instance.
    Further, each feasible schedule corresponds to a selection $S\subseteq V$ of $\abs{S}=k$ vertices, such that no pair of vertices in $S$ share an edge, i.e $(j,k) \in S^2 \Rightarrow (j,k) \notin E$. 
    
\subsection{The Constraint Graph of \fjsp{}}\label{subsec:fjsp-constraint-graph}
    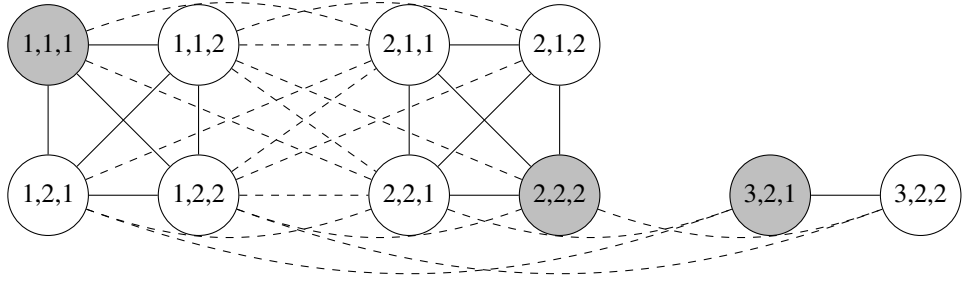
\begin{figure*}
        \centering
        \input{Figures/fjsp.tikz}
        \caption{A visualization of an example constraint-graph for the \fjsp{}.
        Here the instance is given by $J=\{J_1,J_2\}$, $O=\{O_{1,1},O_{1,2},O_{2,1}\}$, $M=\{M_1, M_2\}$, $M_{1,1}=M_{1,2}=M$, $M_{2,1}=\{M_2\}$ and $d_{j,o,m} = 1$ for all $j$, $o$, $m$.
        The nodes are labeled as $i$, $m$, $t$, where $i=1$ denotes $O_{1,1}$, $i=2$ denote $O_{1,2}$ and $i=3$ denotes $O_{2,1}$.
        Due to the assignment constraint, each operation corresponds to a complete subgraph, where the edges are drawn solid.
        The rest of the edges, corresponding to other constraints, are drawn dashed.
        The nodes highlighted in grey correspond to a feasible solution of this \fjsp{} instance.
        Note that no marked nodes are connected with an edge.
        }
        \label{fig:fjsp_graph}
    \end{figure*}
    Consider a \fjsp{} instance, characterized by the data listed in \Cref{subsec:fjsp}. Denote the set of assignments by $\mathcal{A} \coloneqq O \times M \times T$.
    Note that the constraints \ref{itm:ass_constr}, \ref{itm:ord_constraint} and \ref{itm:mach_constr} are all formulated in terms of pairs of assignments.
    We denote these sets of conflicting pairs with respect to assignment, order and machine constraints by

    \begin{align}
        \Omega^{(A)} &= \set*{
        a,a' \in \mathcal{A}|
            a,a' \text{ are in conflict with \ref{itm:ass_constr}}
        },   \\
        \Omega^{(O)} &= \set*{
        a,a' \in \mathcal{A}|
            a,a' \text{ are in conflict with \ref{itm:ord_constraint}}
        },   \\
        \Omega^{(M)} &= \set*{
        a,a' \in \mathcal{A}|
            a,a' \text{ are in conflict with \ref{itm:mach_constr}}
        }. 
    \end{align}
    
    %For a precise definition of the constraint graph for an \fjsp-instance given by $J$, ${O = \bigcup_{J_j\in J} O_j}$, $M$, $\set*{M_{j,o} | O_{j,o}~\in~O}$ and $\set*{d_{j,o,m}|O_{j,o}~\in~O,\,~M_m~\in~M}$ define the following three sets ${\Omega^{(A)} \coloneqq \bigcup_{O_{j,o} \in O} \Omega^{(A)}_{j,o}} $, ${\Omega^{(O)}\coloneqq \bigcup_{J_j\in J}\Omega^{(O)}_j} $, and $ {\Omega^{(M)}=\bigcup_{M_m\in M}\Omega^{(M)}_m}$, of assignment combinations in conflict with the assignment constraint ($\Omega^{(A)}$), order constraint ($\Omega^{(O)}$) or machine constraint ($\Omega^{(M)}$), where 
    %\begin{align}
    %    \Omega^{(A)}_{j,o} &= \set*{
    %    \left(\of{O_{j,o}, M_m, T_t},\of{O_{j,o}, M_{m\primed}, %T_{t\primed}}\right)|
    %        m\neq m\primed \lor t\neq t\primed
    %    }   \\
    %    \Omega^{(O)}_j &= \set*{
    %    \left(\of{O_{j,o}, M_{m}, T_{t}},\of{O_{j,o\primed}, M_{m\primed}, T_{t\primed}}\right) |
    %        o<o\primed \colon\: t+d_{j,o,m}> t\primed
    %    }   \\
    %    \Omega^{(M)}_m &= \set*{
    %    \left(\of{O_{j,o}, M_{m}, T_{t}},\of{O_{j\primed,o\primed}, M_{m}, T_{t\primed}}\right) |
    %    t\leq t\primed < t+d_{j,o,m}
    %    },
    %\end{align}
    %for $O_{j,o}\in O$, $J_j\in J$ and $M_m\in M$.
    Then the construction of the constraint graph described above results in $G=(V, E)$ with 
    \begin{align}
        V &= \set*{\iota\big(O_{j,o},M_m,T_t\big)\in  [N]  | M_m \in M_{j,o}}\,,     \\
        E &= \set*{(j,k) | \cof{\iota^{-1}(j),\iota^{-1}(k)} \in \Omega^{(\text{total})}}\,,
    \end{align}
    where $\Omega^{(\text{total})} \coloneqq \Omega^{(A)} \cup \Omega^{(O)}\cup \Omega^{(M)}$  and $\iota$ is the bijective enumeration of the set of all schedules.
    The number $k$ of vertices to be colored is the total number of operations $k=\abs{O}$.
    \Cref{fig:fjsp_graph} shows the constraint graph for a small instance of the \fjsp{} together with a feasible solution.

\subsection{Graph-controlled Permutation Mixers}\label{subsec:graph-controlled permuation mixers}%
    In the following, we construct mixing operators for the CM-QAOA for the flexible job shop problem making use of the constraint graph formalism. More concretely, an $N$-bit string $x \in \{0,1\}^N$ encoding a schedule $S$ is equivalent to marking $\abs{x}$ vertices of the $N$-vertex constraint graph, where $\abs{x} = \sum_{j}x_j$ is the \emph{Hamming weight} of $x$. The validity of the bit string $x$ corresponds to the property that 
    
\begin{itemize}
    \item there are no adjacent marked vertices in $G$: 
    \begin{equation}\label{eq:adjacent_nodes}
        (j,k)\in E \imply \lnot(x_j \land x_k)
    \end{equation}
    and
    \item exactly $k$ vertices are marked: 
    \begin{equation}\label{eq:hamming_weight}
     \abs{x} = k.   
    \end{equation}
\end{itemize}

    Define the set $F$ of feasible solutions satisfying these constraints as 
        \begin{equation}
            F = \set*{x\in\cof{0,1}^N | x\text{ satisfies \cref{eq:adjacent_nodes} and \cref{eq:hamming_weight}}}
            \label{eq:graph_constraints}
        \end{equation}
    Consider a quantum register of $N$ qubits in $(\C^2)^{\otimes N}$ and denote its computational basis by $\set*{\ket{\vb{x}} | \vb{x}\in \cof{0,1}^N}$. We define the \emph{feasible subspace} to be $\mathcal F = \Span_\C \set*{\ket x | x\in F}$.
    Our goal is to construct a family $B_{\beta} = \set*{\mix\of{\beta} | \beta\in\R}$ of mixing unitaries (cf. \Cref{def:feas_reach}) with respect to the constraints in \cref{eq:adjacent_nodes} and \cref{eq:hamming_weight}.
    To this end, we follow the procedure described by \textcite{hadfield_diss}.
    We start by constructing local and classical mixing rules and continue by translating those to unitary operators.\\
    Starting with a feasible schedule $x$, we want the Hamming weight $\abs{x}$ to be invariant under application of the mixing operators.
    This is ensured by only applying permutations $\pi\in S^N$ to the indices of a bit string, i.e. $\pi(\vb{x}) \mDef \of{x_{\pi^{-1}(1)},\dots,x_{\pi^{-1}(N)}}$.
    The corresponding permutation unitaries
            \begin{equation}
                U_{\pi}\ket{\vb{x}} = \ket{\pi\of{\vb{x}}}
            \end{equation}
    can be extended to a family of unitaries
    \begin{equation}\label{eq:fractional_permutation_unitary}
        U_\pi(\beta) = \cos\of{\beta} \OneOp - \imun \sin\of{\beta}U_\pi
    \end{equation}
    which preserves the constraint \cref{eq:hamming_weight}.
%    While this definition is very similar to and motivated by the fractional \textsc{SWAP} gate $\cos\of{\beta} \OneOp - \imun \sin\of{\beta}\mathrm{\textsc{SWAP}}$, there is an important difference.
%    Other than the fractional \textsc{SWAP}, $U_\pi(\beta)$ is not a 1-parameter group of unitaries, meaning that in general $U_\pi(\alpha +\beta) \neq U_\pi(\alpha)U_\pi(\beta)$.
%    This is due to the fact that in general $U_\pi^2 \neq \OneOp$.
%    Note that for $\beta = (2n+1)\frac{\pi}{2}$ with $n\in\N$, $U_\pi(\beta)\propto U_\pi$ is just the permutation unitary up to a global phase.
    In order to preserve the constraint \cref{eq:adjacent_nodes} we check whether a permutation will preserve the constraint via a boolean function $\chi\colon\:\cof{0,1}^k\to \cof{0,1}$ which decides if an input bit string (with an auxiliary bit attached) is feasible or not.
    Given a unitary operator $U$ acting on $k'$ qubits, we define the $\chi$-controlled unitary operator (see \cite[p.145]{hadfield_diss}) $\Lambda_{\chi}\of{U}$ acting on $k+k\primed$ qubits by the linear extension of
        \begin{equation}\label{eq:controlled_unitary}
            \Lambda_{\chi}\of{U}\ket{x}\ket{y} = \begin{cases}
                \ket{x}\ket{y}, & \chi\of{y}=0\\ (U\ket{x})\ket{y}, & \chi\of{y}=1
            \end{cases}
        \end{equation}
    defined for computational basis states $\ket{x}$ and $\ket{y}$ given by $x\in \cof{0,1}^k, y\in \cof{0,1}^{k'}$. 
    If we find a function $\chi_{\pi}$ for each $\pi$ such that $\chi_{\pi}(x)=1$ if and only if $\pi(x)$ satisfies \cref{eq:graph_constraints}, we can use it as the control for the permutation unitary $U_{\pi}(\beta)$. 
    Denote the neighborhood of a vertex $j\in V$ by $\nbhd\of{j} = \set*{k\in V | (j,k)\in E}$.
    If $x_j=0$ there will be no conflict with the constraint regarding $x_{\pi(j)}$ even if other nodes in $\nbhd(\pi(j))$ are 1.
    On the other hand if $x_j = 1$ all neighbors of $\pi(j)$ need to be 0 after the permutation is applied in order for $\pi$ not to raise a conflict.
    Hence, to prevent conflicts regarding node $j$ we compute
        \begin{equation}\label{eq:single_classical_control}
            \chi_\pi^j\of{x} = \lnot(x_j) \lor \of{\bigwedge_{l\in\nbhd(\pi(j))}\lnot(x_{\pi\inv(l)})}.
        \end{equation}

    Now we can define the complete classical control predicate
        \begin{equation}\label{eq:full_classical_control}
            \begin{aligned}
				\chi_\pi\of{x} 	&= \bigwedge_{j=1}^N \chi_\pi^j\of{x}.
			\end{aligned}
        \end{equation}

        \begin{definition}[Partial graph-controlled permutation mixers]\label{def:partial_mixers}
            For every permutation $\pi\in S^N$ define a \emph{partial graph-controlled permutation mixer}
            \begin{equation}\label{eq:partial_mixers}
                % {\mix}_{\pi}(\beta) =
                \Lambda_{\chi_{\pi}}\of{\tilde U_{\pi}\of{\beta}},
            \end{equation}
            where $\chi_{\pi}$ is the control function defined in \cref{eq:full_classical_control} and 
            \begin{equation}\label{eq:fractional_permutation_unitary_altered}
                \tilde U_\pi\of{\beta} \mDef \cos{\beta}(\OneOp\otimes\OneOp) - \imun \sin{\beta}(X\otimes U_\pi)
            \end{equation}
            is a version of \cref{eq:fractional_permutation_unitary} which is controlled on an auxiliary qubit $\ket{c}$, whose purpose is explained below, cf. \Cref{fig:control_circuit}.
        \end{definition}

   We present one possible approach for implementing the partial graph-controlled permutation mixers \eqref{eq:partial_mixers}.
    We assume that an implementation of $\tilde U_\pi\of{\beta}$ controlled by a single qubit (see below) is known and focus primarily on the implementation of the graph-control.
    \begin{center}\label{fig:controlled_U_tilde}
        \begin{quantikz}
            \qw & \ctrl{1} & \qw\\
            \qw & \gate{\tilde U_\pi\of{\beta}} & \qw
        \end{quantikz}
    \end{center}

    We explicitly implement the logic of $\chi_\pi$ \eqref{eq:full_classical_control} using additional $3N+2$ auxiliary qubits.
    Take some $x\in \cof{0,1}^N$ and consider $3$ $N$-qubits auxiliary quantum registers $y, a, b$ with initial states $\ket{0}_y$, $\ket{0}_a$ and $\ket{0}_b$ and one-qubit registers with initial states $\ket{0}_c$ and $\ket{0}_z$ of one qubit each.
    % If not stated otherwise, all auxiliary qubits will be assumed to be instantiated in $\ket 0$.
    The register $y$ will contain information on the main working register $x$.
    Each of the remaining auxiliary registers will represent truth values of intermediate results in the calculation of $\chi_{\pi}\of{y}$. An overview of the purpose of each register can be found in  \cref{table:registers}. 
    Following the construction of sequential mixers presented in \cite{hadfield_diss} we state the following definition for the mixer operators $\mix(\beta)$.
    \begin{definition}[Graph-controlled permutation mixers]\label{def:permutation_mixers}
    Define the sequential mixer
            \begin{equation}\label{eq:permutation_mixers}
                \mix(\beta) \coloneqq \left( \prod_{\pi\in P \subset S^N}\mix_{\pi}(\beta)\right) A \,,
            \end{equation} 
            where $A$ and $B_\pi\of{\beta}$ defined as in \Cref{fig:control_circuit} and $P$ is a suitable subset of permutations; a natural choice is presented in \Cref{lemma:transpositions}. 
            More precisely, in order to implement $\mix (\beta)$ apply \dots
        \begin{enumerate}[label=(\alph{*})]
        \item \dots $A$ as defined in \Cref{fig:control_circuit}, i.e. for each $j=1,\dots,N$ apply $\textsc{CNOT}_{j;j}\ket{\psi}_x\otimes\ket{0}_y$, where $\ket{\psi}_x$ is a superposition of feasible states in the main register. This step copies an unmodified version of the working register into the $y$ register.
        \item \dots a Pauli $X_j$ on each qubit $b_j$, $j=1,\dots,N$ in the register $b$, flipping the qubits in the $b$ register.
        \item \dots $C_j\of{\pi}$ for each $j=1,\dots, N$ as defined in \cref{eq:control_circuits}.
        This step first uses each qubit $a_j$ in the $a$ register to store whether there are any neighbors of $\pi(j)$ which have the value $1$ in the constraint graph corresponding to the schedule which applying $\pi$ would yield. Then the bit $b_j$ is flipped back to $0$ if $a_j$ has the value $0$ and $y_j$ has the value $1$, i.e. if a conflict (two neighboring vertices in the constraint graph both have value $1$) regarding the $j$-th variable arises. This implements the logic in \cref{eq:single_classical_control}.
        %, i.e. apply $\textsc{CNOT}_{\lnot l_1,\dots,\lnot l_{d\of{j}};j}\ket{y}\otimes \ket{a}$, for $\nbhd\of{\pi\of j} = \cof{\pi^{-1} \of{l_1},\dots,\pi^{-1} \of{l_{d\of{j}}}}$ and $\textsc{CNOT}_{j,\lnot j;j}\ket{y}\otimes \ket{a}\otimes\ket{b}$.
        
        \item \dots a single multiqubit $\textsc{CNOT}_{1,\dots,N,\lnot N+1; 1}$ applied to the registers $b$ and $c$ (control qubits) and $z$. This stores in the register $z$ whether all qubits in $b_j$ have the value $1$, i.e. if applying the permutation would yield a feasible schedule. This is the implementation of \cref{eq:full_classical_control}.
        \item \dots a single qubit controlled $\tilde U_\pi\of{\beta}$ as defined in \cref{eq:fractional_permutation_unitary_altered}, applied to the registers $c$ and $x$ with $\ket{z}$ as control qubit.
        This applies the parameterized permutation, conditioned on register $z$ (which indicates whether the permutation is feasible) and register $c$, which ensures that the permutation is only applied to the original $\ket{\psi}_x$ portion of the working register.
        \item Uncompute the auxiliary registers as depicted in \Cref{fig:control_circuit}, i.e. apply the inverse operator of (b)-(e).
        \item Repeat (b)-(f) for every $\pi\in P\subset S^N$.
    \end{enumerate}

     \begin{table}[ht!]
    \caption{Quantum registers and their purpose} 
    \begin{center}
    \begin{tabular}{ |c|p{0.55\linewidth}|c| }
     \hline
     $y$ & Aux. register: Copy of the initial state before the application of any permutations in the main register & $N$ qubits \\ 
     \hline
     $a$ & Aux. register: For each variable $x_j$: would rescheduling according to $\pi \in S_N$ lead to the assignment $x_{l}=0$ for every neighbor $l$ of $\pi(j)$? & $N$ qubits\\ 
     \hline
     $b$ & Aux. register: If $a_j=1$ (else set $b_j:=1$): would rescheduling according to $\pi \in S_N$ lead to the assignment $x_{\pi{j}}=1$? & $N$ qubits\\
     \hline
     $z$ & Aux. register: Are all of the variables in register $b$ equal to $1$  & $1$ qubit \\
     \hline
     $c$ &  Aux. register: One variable ensuring that each layer only applies permutations to the original state & $1$ qubit \\
     \hline
     $x$ & Main register: Modified by every time we apply $B_\pi(\beta)$ for a permutation yielding a legal schedule  & $N$ qubits \\
     \hline
    \end{tabular}
    \label{table:registers} 
    \end{center}
    \end{table}
    
    Here $\textsc{CNOT}_{l_1,\dots,l_n;k}$ denotes the multi qubit \textsc{CNOT} acting on qubits $l_1,\dots,l_n$ as control and on qubit $k$ as target.
    A \glqq{}$\lnot$\grqq{} in an expression like $\textsc{CNOT}_{i,\lnot j;k}$ denotes that the $j$-th qubit should act as negated control qubit, while the $i$-th qubit should act as normal control.
    The number of neighbors of a node $j$ in the graph is denoted by $d(j)$. 
    % Negated control means, that the target state gets flipped, if all control qubits are in the $0$ state instead of the $1$ state.
    The whole procedure is illustrated in \Cref{fig:control_circuit},
       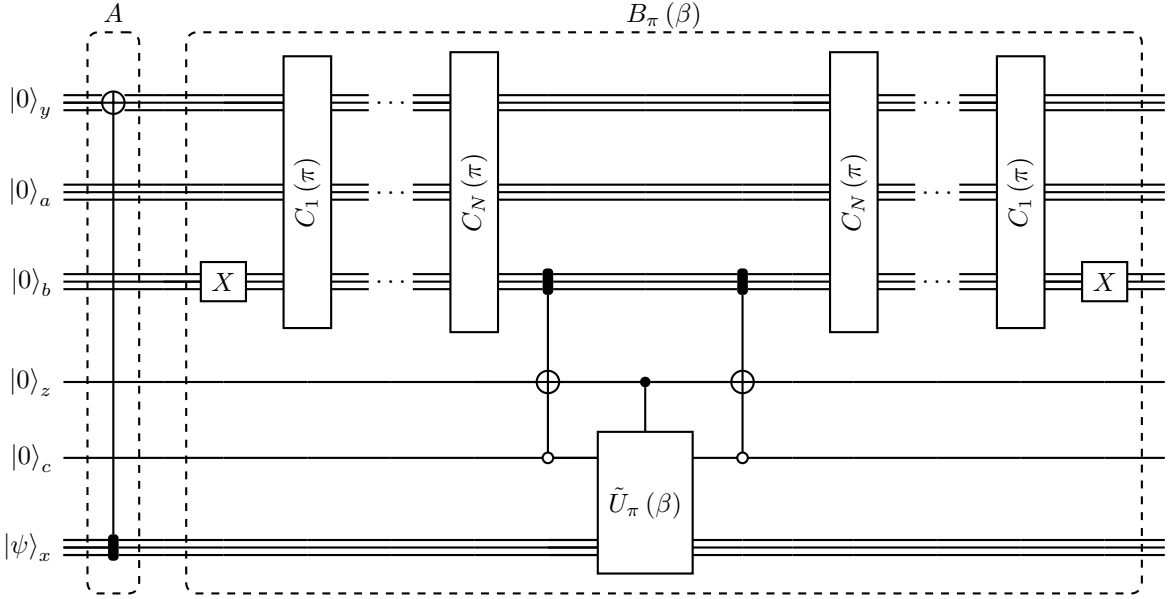
\begin{figure*}
        \centering
        \begin{quantikz}
            \lstick{$\ket{0}_y$}\qwbund{}& \targ{}\qwbund{}\gategroup[wires=6, steps=1, style={dashed,rounded corners, inner xsep=2pt}]{$A$} 
            &\qwbund{}   &\qwbund{}\gategroup[wires=6, steps=12,style={dashed,rounded corners, inner xsep=2pt}]{$\mix_\pi\of{\beta}$}%
                   & \gate[wires=3, bundle={2,3} ]{\rotatebox{90}{$C_1\of{\pi}$}} \qwbund{} %
               &  \ \ldots\ \qwbund{} %
              & \gate[wires=3,bundle={2,3} ]{\rotatebox{90}{$C_N\of{\pi}$}}\qwbund{}          & \qwbund{}     & \qwbund{}                     &\qwbund{} &\qwbund{}
          & \gate[wires=3, bundle={2,3} ]{\rotatebox{90}{$C_N\of{\pi}$}} \qwbund{} %
            &  \ \ldots\ \qwbund{} %
            &\gate[wires=3,bundle={2,3} ]{\rotatebox{90}{$C_1\of{\pi}$}}\qwbund{}     &\qwbund{} &\qwbund{}             \\
            \lstick{$\ket{0}_a$} &\qwbund{}                &  \qwbund{} &\qwbund{}          &         &\ \ldots\ \qwbund{}&         & \qwbund{}     & \qwbund{}                     &\qwbund{} &\qwbund{}&         &\ \ldots\ \qwbund{}&          &\qwbund{} &\qwbund{}
            \\
            \lstick{$\ket{0}_b$} &\qwbund{}                &\qwbund{}   &  \gate{X}\qwbund{}&         &\ \ldots\ \qwbund{}&         & \ctrlbundle{1}& \qwbund{}                     & \ctrlbundle{1} &\qwbund{}&         &\ \ldots\ \qwbund{}&         &  \gate{X}\qwbund{}&\qwbund{}
            \\
            \lstick{$\ket{0}_z$}        &\qw                      &\qw         &\qw                &\qw      & \qw      &\qw      &\targ{}\qw   &\ctrl{1}\qw                         &\targ{}\qw    &\qw     &\qw         &\qw                &\qw       &\qw &\qw      
            \\
            \lstick{$\ket{0}_c$} &\qw                      &  \qw       &\qw                &  \qw    &   \qw    & \qw     & \octrl{-1}       & \gate[wires=2, bundle={2}]{\tilde U_\pi\of{\beta}}\qw                       & \octrl{-1} &\qw      &\qw      &\qw    & \qw       &\qw  &\qw     
            \\
            \lstick{$\ket{\psi}_x$}  &\ctrlbundle{-5}\qwbund{} &\qwbund{}   &  \qwbund{}        &\qwbund{}& \qwbund{}&\qwbund{}& \qwbund{}     & \qwbund{}&\qwbund{}& \qwbund{}    &\qwbund{}& \qwbund{}    &\qwbund{}&\qwbund{}&\qwbund{}
        \end{quantikz}
        \caption{
        The quantum circuit of $B_{\pi}\of{\beta}A$ used for graph-controlled permutation mixers $\mix(\beta)$ from \Cref{def:permutation_mixers}. The $C_j\of{\pi}$ are defined in \cref{eq:control_circuits} and $\tilde U_\pi\of\beta$ is defined in \cref{eq:fractional_permutation_unitary_altered}. The auxiliary registers are initialized in the zero state, the working register contains a superposition of feasible states $\ket{\psi}_x$.
        }
        \label{fig:control_circuit}
    \end{figure*}
    where $C_j\of{\pi}$ is defined in \cref{eq:control_circuits}.
    \begin{equation}\label{eq:control_circuits}
    %\begin{split}
            \begin{quantikz}
                \lstick{$y$} \qwbund& \gate[wires=3,bundle={2,3} ]{C_j\of{\pi}} \qwbund & \qwbund &   \\
                \lstick{$a$} \qwbund&                                  \qwbund & \qwbund &   \\
                \lstick{$b$} \qwbund&                                  \qwbund & \qwbund &   
                                     % &&&
            \end{quantikz} 
            \coloneqq  
            \begin{quantikz}
                \lstick{$y_{j}$}             \qw &   \qw         &    \qw                & \ctrl{2}      &   \qw \\
                \lstick{$y_{P_j}$}   \qwbundle{}  &   \qwbundle{d\of{\pi\of{j}}}     &    \octrl{1}    & \qwbundle{}   &   \qwbundle{} \\
                \lstick{$a_j$}               \qw &   \qw         &    \targ{}            & \octrl{1} &   \qw \\
                \lstick{$b_j$}               \qw &   \qw         &    \qw                & \targ{}   &   \qw 
            \end{quantikz}\,,
    %\end{split}
    \end{equation}
    where $P_j \coloneqq \pi^{-1}\of{\nbhd\of{\pi(j)}}$.
     \end{definition}
    Note that the uncomputation leads to the states $\ket{0}$ only in registers $a$ and $b$. For the auxiliary registers $y$, $c$ and $z$, we require a fresh register for each layer of mixers. Thus, this protocol requires $2N + (2+N)k$ auxiliary registers in total for $N$ logical qubits and $k$ layers in \qaoa.
    
    There are several options to choose a suitable set of permutations $P$. Considering the entire set $S^N$ fulfills all requirements of \Cref{def:feas_reach} but is practically infeasible since it would require checking $|S^N|= N!$, i.e. factorially many permutations. 
    The following lemma states that it is sufficient to restrict to transpositions.
    \begin{lemma}\label{lemma:transpositions}
    Given feasible schedules $x,y \in F \subset \{0,1\}^N$, there are transpositions $\tau_1,...,\tau_k \in S_N$ such that
    \begin{equation}
        \left(\prod_{i=1}^{k}\tau_i\right)(x)=y \text{ and } \left(\prod_{i=1}^{l}\tau_i\right)(x) \in F 
    \end{equation}
   for each $1 \leq l \leq k$. 
    \end{lemma}
    \begin{proof}
    Assume that the makespans $C(x)$, $C(y)$ of schedules $x$ and $y$ satisfy $\max(C(x), C(y))<\frac{1}{2}n_T$, else choose a higher cutoff for number of steps $n_T$.\\
    Let $j \in J$ and $o \in O_j$ specify a job and an operation. Without loss of generality, we assume that the bijective enumeration function $\iota: O \times M \times T \rightarrow \{1,\dots,N\}$ is built such that $\iota(o,m,2t) = 2\iota(o,m,t)$. Define indices $b_{j,o}, c_{j,o} \in \{1,\dots,N\}$ such that
    \begin{equation}\label{eq:b_jo}
    \iota^{-1}(b_{j,o}) = (j,o,m,t) \text{ for some } m \text{, } t \text{ and } x_{b_{j,o}}=1   
    \end{equation}
    and 
    \begin{equation}\label{eq:c_jo}
    \iota^{-1}(c_{j,o}) = (j,o,m',t') \text{ for some } m' \text{, } t' \text{ and } y_{c_{j,o}}=1. 
    \end{equation}
    In other words $b_{j,o}$ and $c_{j,o}$ are the indices corresponding to operation $o$ of job $j$ whose entry in the schedule $x$ resp. $y$ is $1$. As every feasible schedule assigns an operation to exactly one machine and starting time, the properties \eqref{eq:b_jo} and \eqref{eq:c_jo} characterize $b_{j,o}$ and $c_{j,o}$ uniquely.
    
    Define transpositions $\tau_{j,o}^{1 \rightarrow 2}, \tau_{j,o}^{2 \rightarrow 1} \in S_N$ via
    \begin{equation}
        \tau_{j,o}^{1 \rightarrow 2}(a) = 
        \begin{cases}
            2a & \text{ if } a = b_{j,o}\\
            \frac{a}{2} & \text{ if } \frac{a}{2} = b_{j,o}\\
            a & \text{ otherwise}
        \end{cases}
    \end{equation}
    and
    \begin{equation}
        \tau_{j,o}^{2 \rightarrow 1}(a) = 
        \begin{cases}
            c_{j,o} & \text{ if } \frac{a}{2} = b_{j,o}\\
            2b_{j,o} & \text{ if } a = c_{j,o}\\
            a & \text{ otherwise.}
        \end{cases}
    \end{equation}
$\tau_{j,o}^{1 \rightarrow 2}$ is the transposition which shifts the assignment of operation $(j,o)$ according to schedule $x$ into the second half of the available time window. The transposition $\tau_{j,o}^{2 \rightarrow 1}$ shifts the assignment of operation $(j,o)$ back into the first half of the available time window, but into the position specified by the goal schedule $y$.
Now, we define an order in which we apply the transpositions $\tau_{j,o}^{1 \rightarrow 2}$ and $\tau_{j,o}^{2 \rightarrow 1}$ for each job and operation $(j,o)$. We first apply $\tau_{j,o}^{1 \rightarrow 2}$ and within every job $j$, we go from last to first operation, thereby ensuring that the order constraint is satisfied at all times.
\begin{equation}
    \pi := \prod_{j \in J} \prod_{o=1}^{O_j} \tau_{j,o}^{1 \rightarrow 2}\,.
\end{equation}
In the second step, we apply the transpositions $\tau_{j,o}^{2 \rightarrow 1}$, this time going from first to last operation within each job.
\begin{equation}
    \rho := \prod_{j \in J} \prod_{o=O_j}^{1} \tau_{j,o}^{2 \rightarrow 1}\,.
\end{equation}
In total, we get the desired result
\begin{equation}
    \rho \pi (x) = y\,.
\end{equation}
\end{proof}
    
Now we will show that the mixing operators given in \Cref{def:permutation_mixers} preserve feasibility with respect to the constraints \eqref{eq:graph_constraints}.
This implies that, when starting in a feasible state and running \qaoa, the evolution stays in the feasible subspace of the given problem.
There will be no overlaps with non-feasible states and the algorithm is exclusively focused on optimizing with regard to a given goal function. In the following, we denote the Hilbert space of the auxiliary registers by $\mathcal{H}_{\text{aux}}\coloneqq (\C^{2})^{\otimes (3N+2)}$ and $\ket{0}_{\text{aux}}\coloneqq \bigotimes_{i=1}^{3N+2}\ket{0}$.
        \begin{theorem}\label{thm:feasibility}(\textit{Preserving feasibility}).
            The mixer operators $\mix(\beta)$ given in \Cref{def:permutation_mixers} preserve feasibility with regard to the constraints given in \cref{eq:graph_constraints}, i.e. for every $x\in F$, $\beta \in \R$ we have $$\mix(\beta)\left(\ket{x}\otimes \ket{0}_{\text{aux}}\right) \in\feas\otimes \mathcal{H}_{\text{aux}}\,.$$
        \end{theorem}
        \begin{proof}
            Given any list of permutations $\pi_1,...,\pi_n \in S_N$, the quantum state after applying $\prod_i B_{\pi_i}(\beta)A$ to a feasible initial basis state $\ket{x}$ with $x\in F$ is given by 
            \begin{align*}
                &\left( \prod_{i=1}^n B_{\pi_i}(\beta) \right) A \left(\ket{x}_x \otimes \ket{0}_{\text{aux}}\right) = \\
                &\left( \prod_{i=1}^n B_{\pi_i}(\beta) \right) \left(\ket{x}_x \otimes \ket{x}_y \otimes \ket{0}_{a,b,z,c}\right) = \\
                &\alpha \left(\ket{x}_x \otimes \ket{x}_y \otimes \ket{0}_{a,b,z,c}\right) + \\
                &\sum_{i=1}^n \alpha_i \left(\ket{\pi_i(x)}_x \otimes \ket{x}_y \otimes\ket{0}_{a,b}\otimes \ket{\text{trash}(i)}_{z,c}\right),
            \end{align*}
            where $\alpha_i = 0$ iff $\pi_i(x)$ is not a feasible solution. This is ensured by construction of the control logic with auxiliary registers $a,b,y$ and $z$. The control on auxiliary register $c$ ensures by construction that a permutation $\pi_i$ is only ever applied to the initial $\ket{x}$ portion of the superposition in the main register, cf. \Cref{fig:control_circuit}.
        \end{proof}
The second criterion $\mix$ needs to meet in order for being a suitable mixer is that every feasible basis state in $\feas$ can be reached by a number of repetitions.
    This is necessary for the phase-separator to have access to all possible solutions in order to find the best possible one.
    In fact, $\mix$ as given in \Cref{def:permutation_mixers} does have this property. We define $\ket{0}_{\text{aux,k}}\coloneqq \bigotimes_{i=1}^{2N+(N+2)k}\ket{0}$.

    \begin{theorem}\label{thm:explorability}(\textit{Exploring the feasible subspace}).
        Let $x\in F$, denote the initial state in the auxiliary registers as $\ket{0}_{\text{aux},k}\in \mathcal{H}_{\text{aux},k}$, the mixer $\mix(\beta)$ as given in \Cref{def:permutation_mixers} and $\ket{\psi_{\text{mix}, k, \beta}(x)} \coloneqq \left(\mix(\beta)\right)^k\left(\ket{x}\otimes \ket{0}_{\text{aux,k}}\right)$. Further define $\rho^x_{\text{mix}, k, \beta}\coloneqq \Tr_{\text{aux}}{\left[\ketbra{\psi_{\text{mix}, k, \beta}(x)}{\psi_{\text{mix}, k, \beta}(x)}\right]}$. Then for every $x,x\primed\in F$ there are some $k\in\N$ and $\beta\in \R$  s.t. 
        $$\Tr\left[\ketbra{x\primed }{x\primed }\rho^x_{\text{mix}, k, \beta}\right]>0\,.$$
    \end{theorem}
    \begin{proof}
        Let $x,x\primed \in F$. Denote the transpositions provided by \Cref{lemma:transpositions} by $\tau_1,...,\tau_k \in S_N$ and recall their properties 
        \begin{equation}\label{eq: taus in proof of explorability}
        \left(\prod_{i=1}^{k}\tau_i\right)(x)=x' \text{ and } \left(\prod_{i=1}^{l}\tau_i\right)(x) \in F  
    \end{equation}
     for each $ 1 \leq l \leq k$. Then 
    \begin{align*}
        &\ket{\psi_{\text{mix}, k, \beta}(x)} = \left(\mix(\beta)\right)^k  \left(\ket{x}\otimes \ket{0}_{\text{aux,k}}\right) \\
        &= \alpha \ket{\tau_k\cdots \tau_1 (x)}\otimes  \ket{0}_{a,b}\otimes \ket{\text{trash}}^{\otimes k}_{y,z,c} + \ket{\text{rest}} \\
        & = \alpha \ket{x'}\otimes  \ket{0}_{a,b}\otimes \ket{\text{trash}}^{\otimes k}_{y,z,c} + \ket{\text{rest}},
    \end{align*}
    where $\alpha =\prod_{j=1}^k \alpha_j$, $\alpha_j>0$ being the probability that we apply $\tau_j$ in the $j$-th layer. The remaining sum of feasible states $\ket{\text{rest}}$ may possibly contain a summand with $\ket{x'}$, so we choose $\beta$ such that it does not cancel out the coefficient $\alpha$. Denote the state in the auxiliary registers corresponding to $x'$ by $\ket{\psi (x')}_{\text{aux}} \coloneqq \ket{0}_{a,b}\otimes \ket{\text{trash}}^{\otimes k}_{y,z,c}$ as derived above. Then the probability of obtaining $\ket{\psi (x')}_{\text{aux}}$ when measuring the auxiliary registers of $\ket{\psi_{\text{mix}, k, \beta}(x)}$ in the computational basis is given by $P(\ket{\psi (x')}_{\text{aux}})>0$. Thus, for a suitable choice of $\beta$ and $k$ as chosen in \cref{eq: taus in proof of explorability} we have
        \begin{align*}
            \Tr\left[\ketbra{x\primed }{x\primed }\rho^x_{\text{mix}, k, \beta}\right]=
            P\left(\ket{\psi (x')}_{\text{aux}}\right)\alpha^2 > 0\,.
        \end{align*}
    \end{proof}

\section{Discussion and Outlook}\label{sec:discussion}
    We have introduced mixer operators for the \fjsp{} and have proven that they preserve feasibility and explore the entire feasible subspace. 
    While the considered constraint graph model is the same as the one considered by \textcite{hannover_paper}, the general approach is very different. 
    The work by \textcite{hannover_paper} has implicitly shown, that for the \fjsp{} permutation unitaries are not sufficient to construct mixing unitaries in the sense of \textcite{hadfield_diss}, when no auxiliary qubits are used.
    In this work, we employ several auxiliary registers in order to ensure that the emerging mixing unitaries preserve the feasible subspace.
    In general, the described control protocol may be applied to all problems that can be modelled by a constraint graph. Additionally, our construction includes a few tricks which can be reused in other protocols, such as the auxiliary register $c$ which stores the information whether an operator has been applied to the quantum state, thus avoiding that a product of permutations is applied. The direct implementation of constraints in the operators $C_i$ can be used whenever constraints admit a boolean formulation as in \cref{eq:single_classical_control}.
    \newline
    The practicality of the presented mixing operators can be debated critically. The direct implementation $C_1 \dots C_N$ of the control logic adds an overhead of $\mathcal O \of{N^3}$ Toffoli, \textsc CNOT and single qubit gates to the single qubit controlled $U_\pi$. In total, the mixer operator $\mix\of{\beta}$ applies at least $\mathcal{O}(N^2) \of{\mathcal O\of{N^3} + n_p} = \mathcal O\of{N^5} + n_p \mathcal O\of{N^2}$ Toffoli, \textsc CNOT and single qubit gates, where $n_p = \min\set*{n_\pi | \pi \in S^N}$ denotes the minimal number of basic gates necessary to implement the single qubit controlled $U_\pi$. This results in a polynomial gate count, which we suspect, however, to be quite demanding on current quantum systems, considering their limited coherence times. 
    A direction for further research could be to exploit the logical structure of specific constraints in order to enforce these in a more efficient way. For certain types of constraints, this has been achieved by \cite{CM_QAOA_Fuchs2022}. Also, one could employ a hybrid approach, where some constraints are enforced via efficiently constructed constraint mixers, while for others the more common penalty term method is kept in place.
    Further, numerical simulations may give more insight on the potential of approaches using constraint mixers.
    Also, a combination of both the direct approach presented in this work and the abstract analysis of the constraint-graph \parencite{hannover_paper} could reduce the number of necessary gates and result in an applicable set of mixing unitaries for solving the \fjsp{} with the \qaoa{}.
    This can be possibly achieved following \cite{Meyer_2023}, which presents a concrete method to design variational circuits which take into account problem symmetries.
    \newline
    Another way to reduce the number of auxiliary qubits as well as the gate complexity of one single $\mix_\pi$ \eqref{eq:partial_mixers}, might be a deeper analysis of the logarithm of the fractional permutation unitaries \eqref{eq:fractional_permutation_unitary}. Finding a general logarithm for these unlocks a different control implementation by considering the exponential of the tensor product of the logarithm and the Hamiltonian simulating the classical control logic.

%% file: Figures/fjsp.tikz
\begin{tikzpicture}[scale=.4]
		\node [style=empty, fill= black!25] (0) at (-5, 5) {1,1,1};
		\node [style=empty] (1) at (0, 5) {1,1,2};
		\node [style=empty] (2) at (-5, 0) {1,2,1};
		\node [style=empty] (3) at (0, 0) {1,2,2};
		\node [style=empty] (4) at (7, 5) {2,1,1};
		\node [style=empty] (5) at (12, 5) {2,1,2};
		\node [style=empty] (6) at (7, 0) {2,2,1};
		\node [style=empty, fill=black!25] (7) at (12, 0) {2,2,2};
		\node [style=empty, fill=black!25] (10) at (19, 0) {3,2,1};
		\node [style=empty] (11) at (24, 0) {3,2,2};
		\draw [style=standard] (4) to (7);
		\draw [style=standard] (6) to (5);
		\draw [style=standard] (0) to (3);
		\draw [style=standard] (2) to (1);
		\draw [style=Other constraint, bend left=20, looseness=1] (0) to (4);
		\draw [style=Other constraint] (0) to (6);
		\draw [style=Other constraint] (1) to (4);
		\draw [style=Other constraint, bend left=20, looseness=1] (1) to (5);
		\draw [style=Other constraint] (1) to (6);
		\draw [style=Other constraint] (1) to (7);
		\draw [style=Other constraint] (2) to (4);
		\draw [style=Other constraint, bend right=20] (2) to (6);
		\draw [style=Other constraint, bend right=20, looseness=1] (2) to (10);
		\draw [style=Other constraint] (3) to (4);
		\draw [style=Other constraint] (3) to (5);
		\draw [style=Other constraint] (3) to (6);
		\draw [style=Other constraint, bend right=20, looseness=1] (3) to (7);
		\draw [style=Other constraint, bend right=20, looseness=1] (3) to (11);
		\draw [style=Other constraint, bend right=20, looseness=1] (6) to (10);
		\draw [style=Other constraint, bend right=20] (7) to (11);
		\draw [style=standard] (0) to (2);
		\draw [style=standard] (2) to (3);
		\draw [style=standard, in=270, out=90] (3) to (1);
		\draw [style=standard] (1) to (0);
		\draw [style=standard] (4) to (5);
		\draw [style=standard] (4) to (6);
		\draw [style=standard] (6) to (7);
		\draw [style=standard] (7) to (5);
		\draw [style=standard] (10) to (11);
\end{tikzpicture}

%% file: HICSScontent/hicss_acknowl.tex
\section*{Acknowledgements}
The authors thank L. Binkowski, J. R. Fin\v zgar, C. B. Mendl, T. Osborne and B. Poggel for fruitful discussions and comments. 
LP and MH acknowledge financial support
by the BMBF project QuBRA.